\documentclass{article}
\usepackage{graphicx} % Required for inserting images

\title{A Tractability Gap Beyond Nim-Sums: It's Hard to Tell Whether a Bunch of Superstars Are Losers}
\author{Kyle Burke \\ Florida Southern College \and Matthew Ferland \\ USC \and Svenja Huntemann \\ Mount Saint Vincent University \and Shang-Hua Teng\thanks{Supported in part by NSF Grant CCF-2308744 and the
Simons Investigator Award from the Simons Foundation. }\\ USC}
\date{May 2023}
\usepackage{paithan}
\usepackage{amsthm}
\usepackage{mathtools}
\usepackage{tikz}
\usetikzlibrary{shapes, positioning, backgrounds, fit, arrows, arrows.meta}
\usepackage{xcolor}
\usepackage{xspace}
\usepackage{enumitem}
\usepackage{dsfont}
\usepackage{hyperref}
\usepackage[noabbrev,capitalise]{cleveref}
\usepackage{bold-extra}

\theoremstyle{definition}
\newtheorem{definition}{Definition}[section]

\theoremstyle{plain}
\newtheorem{theorem}[definition]{Theorem}
\newtheorem{corollary}[definition]{Corollary}
\newtheorem{lemma}[definition]{Lemma}
\newtheorem{conjecture}[definition]{Conjecture}
\newtheorem{proposition}[definition]{Proposition}

\newcommand{\kyle}[1]{{\color{blue} [Kyle: #1]}}

\newcommand{\oneonset}{\ensuremath{S^\text{OneOn}}}
\newcommand{\alloffset}{\ensuremath{S^\text{AllOff}}}

\begin{document}

\maketitle

%In this paper we address a natural question in the intersection of combinatorial games and computational complexity: ``can a sum of simple tepid games in canonical form be intractable?'' Specifically, we consider \textit{superstars}, positions where all options are \textit{nimbers}, %use the recent construction of quasi-nimbers, and the classical construction of super-stars by Conway 
%to illustrate this concept. We prove that disjunctive sums of superstars are NP-hard to solve.  % as well as sums of superstars with a single nimber.  
%This is striking as superstars, first introduced in \textit{Winning Ways}, are only a single move away from nimbers, whose sums can be computed in polynomial time, and this extending Morris' classic results on hot games to tepid games.  %, by the celebrated Sprague-Grandy theory. 
%Our analyses lead to a family of elegant board games with intriguing complexity.
%We also present web-playable versions of the rulesets described within.

\begin{abstract}
In this paper, we address a natural question at the intersection of combinatorial game theory and computational complexity: ``Can a sum of simple \emph{tepid games} in canonical form be intractable?'' To resolve this fundamental question, we consider \textit{superstars}, positions first introduced in \textit{Winning Ways} where all options are \textit{nimbers}. %use the recent construction of quasi-nimbers, and the classical construction of super-stars by Conway 
Extending Morris' classic result with hot games to tepid games, %, by the celebrated Sprague-Grandy theory. 
we prove that disjunctive sums of superstars are %NP-hard 
intractable to solve.  % as well as sums of superstars with a single nimber.  
This is striking as sums of nimbers can be computed in linear time.  
Our analyses also lead to a family of elegant board games with intriguing complexity, for which we present web-playable versions of the rulesets described within.
\end{abstract}

\section{Introduction}
``The whole is greater than the sum of their parts'' is an ancient phrase that particularly exemplifies combinatorial game theory. As an area of mathematics dedicated to analyzing what happens when several games are combined, the field is rich with results both in isolation and with interdisciplinary connections. Indeed, even casually, games are often combined for enjoyment, such as Bughouse (2 simultaneous games of Chess) and Ultimate Tic Tac Toe (9 simultaneous games of Tic Tac Toe).

While several different ways to combine games are studied, the predominant one is the \textit{disjunctive sum of games} following the \textit{normal play} convention. In a disjunctive sum, players alternate turns choosing a single game component, making a move on it, and passing the turn over to their opponent, leaving the other components unchanged.  Under normal play, when a player is unable to make a move (because there are no moves remaining for them in any game), then that player loses.  In other words, the last player to make a move wins.  

The modern forms of both combinatorial game theory and computational complexity theory were born around half a century ago, and there are several results of the latter about the former.  Most relevant to this work, in 1981, Morris demonstrated that a sum of (individually polynomial-time solvable games) is PSPACE-complete \cite{morris1981playing}.  The hard game sums in that reduction have several key requirements.  First, they involve deeply asymmetric games (i.e., games where the moves available to the two players are very different). Second, the games have exponential length (the number of turns). Third, a polynomial number of games are included in the sums.  Finally, most components of the game are \textit{hot games}, meaning that players are incentivized to play first on most games in the sum.

Later results by Yedwab \cite{yedwab1985playing}, Mowes \cite{moews1993some}, and eventually Wolfe \cite{wolfe2002go} improved the reduction by eliminating the exponential length and reducing the branching factor size (to a smaller constant). However, the other two limitations remain.

A more recent result in 2021 \cite{BurkeFerlandTengGrundy} demonstrated that the sum of two tractable symmetric polynomial-length impartial games\footnote{positions in which both players have the same options}, when combined, are PSPACE-complete. This was accomplished using a pair of natural
games known as \ruleset{Undirected Geography}.
Although \ruleset{Undirected Geography} positions can be solved in polynomial time \cite{DBLP:journals/tcs/FraenkelSU93},
 it is \cclass{PSPACE}-complete to determine their values, as shown in \cite{BurkeFerlandTengGrundy}.
 %These graph-theoretical game positions also have ``shallow'' game trees.
% even though the possible values are each equivalent to ``shallow'' \emph{impartial} positions.  %can each be expressed succinctly.  These ``shallow values'' can be summed easily. %was done by having the games outside of what is known as \textit{canonical form}, or the simplest expression of how games interact in sums. 

%\begin{definition}[Impartial]
%    A game position is \emph{impartial} if both players have the same options.  A ruleset is \emph{impartial} if every position possible in that ruleset is impartial.
%\end{definition}

%Impartial rulesets include those such as \ruleset{Nim} and \ruleset{Kayles}.  Rulesets such as \ruleset{Chess} and \ruleset{Go} are not impartial because players can only move or place pieces of their color, so the options are different.  Non-impartial games are known as \emph{partizan}.  The methods for simplifying, evaluating, and adding impartial games in normal play are fully described by \emph{Sprague Grundy Theory}.  This assigns each impartial position a \emph{nim value}, recursively defined using the \emph{mex rule}.

%\begin{definition}[Mex]
%    The \emph{mex} of a set of non-negative integers, $\mex{ s_1, s_2, \ldots, s_k}$ is the smallest non-negative integer that does not appear in the set.  
%\end{definition}

%\begin{itemize}
%    \item Finish S-G theory part here.
%    \item Nimbers
%    \item XOR 
%\end{itemize}

Therefore, the hardness comes from finding that value rather than describing the difficulty of performing the mathematical operation. In fact, these impartial games have a simple polynomial algorithm to identify a winner in a sum if the game is already in its simplest form.

%This paper proceeds on an alternative route from on Morris' results. 
This paper continues the chain of results that Morris started, finding intractible summands with even more shallow game trees than were previously known.  First, instead of hot games, we sum components from a family of \textit{tepid} games.  Tepid is a term based in temperature theory, where the temperature of a game is the number that approximates the incentive to move in a position.  Hotter games have a higher value of potentially-earned moves in the favor of the first player to play on them.  This can be by supplying moves to use later or denying your opponent later free plays.\footnote{For a more thorough description of temperature theory in CGT, we recommend \cite{SiegelCGT:2013}.}  Cold games use up these moves when a player plays on them; their temperature is negative.  Tepid games all have a temperature of zero: playing on them doesn't earn either player any moves but can influence the parity of the current situation.  

The second reason our work continues the chain of intractible sums is that, instead of the more deeply asymmetric games that Morris used, we use a family of \textit{nearly symmetric} games,  which become symmetric after a single move. Third, unlike the result with \ruleset{Undirected Geography}, we use a family of games that are already in \textit{canonical form}, which is to say, directly in the form of their values. And finally, this family of games is deeply related to one whose existence traces back to the birth of the modern theory.

The main family of games we consider are called \textit{superstars}\footnote{There are some inconsistencies with this choice of name, which we discuss in full in  \cref{sec:superstarTerminology}.}.%, which are partizan games with nimber-valued options.
These values naturally occur in the game \ruleset{Paint Can}\cite{SILVA2023113665}, which we discuss in \cref{sec:background}.  We show that for sums of superstars, it is computationally-intractable to determine which player has a winning strategy. %TODO: we have to work this in: "It's hard to tell whether a bunch of superstars are losers."  

%\svenja{Do we want to leave this here or move it to the relevant section?}

\begin{theorem}
A sum of superstars is \cclass{NP}-hard. %\svenja{The tepid is superfluous here. Do we want to keep it?}
\end{theorem}

%The reduction used to prove our main theorem also leads to a nice game which we call \ruleset{Light Switching}.

The paper is structured as follows: In \cref{sec:background} we introduce the necessary concepts from combinatorial game theory. The proof of the main theorem is given in \cref{sec:reduction}. The reduction used to prove our main theorem also leads to a nice new ruleset which we call \ruleset{Blackout}, introduced in \cref{sec:lightswitching}.

\section{Superstars: Theory and \sc{Paint Can}}\label{sec:background}

\subsection{Rising Stars: From Stars to Superstars}

In this section we will give a brief introduction to concepts from combinatorial game theory (CGT) required for this paper. For more information and a rigourous treatment of these topics, see \cite{WinningWays:2001}, \cite{LessonsInPlay:2007}, and \cite{SiegelCGT:2013}. 

In the game {\sc Nim}, played on piles of tokens, the two players take turns choosing a pile and removing any nonzero number of tokens. Under \textit{normal play} the player to pick the last token(s) wins.

\textsc{Nim} is an \emph{impartial game}, meaning one in which the two players have the same possible moves. An impartial game is denoted $G=\{G_1,\ldots,G_n\}$, where $G_1,\ldots, G_n$ are the options the players can move to.

When only a single pile remains in \textsc{Nim}, the current player will simply remove all tokens. But when several piles remain, the optimal move is often to only take some of the tokens. Thus all possible moves on a pile need to be considered when in a sum. To do so, we assign a \textit{value} to each pile which represents the possible moves. An empty pile, thus one in which there are no moves, is given the value $0$. The value of a pile with $n$ tokens is the \textit{nimber} $*n$. For a consistent recursive definition, we think of $0$ as the nimber $*0$. The shorthand $*1=*$ is also generally used.

\begin{definition}
The nimber $*n$ is recursively defined by its options as
\begin{itemize}
    \item $*0=0=\emptyset$ (no available moves);
    \item $*1=*=\{0\}$
    \item $*n=\{0,*,*2,\ldots,*(n-1)\}$
\end{itemize}
\end{definition}

The \emph{(disjunctive) sum} $G_1+\cdots+G_n$ of games $G_1, \ldots, G_n$ is the game in which the players chose a summand (or component), then make a move in it. A Nim position is naturally the disjunctive sum of its separate piles. Many other games also naturally break into components, but we can consider sums of any games in general.

We say that two games are \emph{equal} to each other when one can be replaced with the other in \emph{any} disjunctive sum without changing the winnability. I.e., $A=B$ whenever who wins in $A + X$ is the same as in $B + X$ for any game $X$.

A sum of nimbers is always equal to a single nimber. Finding which nimber it is requires only an XOR sum (also known as \textit{nim sum} in CGT) and therefore is in \cclass{P} (solveable in polynomial time). % (Two game positions $A$ and $B$ are equal if for any other game, $G$, $A + G$ and $B +G$ has the same winnability.)  \svenja{How much do we want to go into detail about equality? Should I define canonical form or only define equality more in general?} \kyle{How's that? }

In {\sc Nim}, both players have the same available moves. For a game where the two players, which we call \textit{Left} and \textit{Right}, have differing moves, we use the notation
\[\{\text{Left's options}\mid\text{Right's options\}}.\]
Such games are called \textit{partizan games}, and they have four \emph{outcome (winnability) classes}. A game in
\begin{itemize}
    \item $\mathcal{N}$ is won by the player that moves first, no matter whether they are Left or Right;
    \item $\mathcal{P}$ is won by the player that moves second, no matter whether they are Left or Right;
    \item $\mathcal{L}$ is won by Left no matter who goes first; and
    \item $\mathcal{R}$ is won by Right no matter who goes first.
\end{itemize}

\begin{definition}
A \textit{superstar} is a game in which all options for Left and Right are nimbers, possibly not all the same.
\end{definition}

A superstar in which the options for both players are the same is a nimber. Even more in general we have the following:

\begin{proposition}[\cite{ONAG:2000}]\label{thm:superstarSimplified}

The superstar
\[\{0,*,\ldots,*(n-1), *x_1,\ldots,*x_k\mid 0,*,\ldots,*(n-1), *y_1,\ldots,*y_l\},\]
where $x_i,y_j>n$ for all $i$ and $j$, is equal to the nimber $*n$.
\end{proposition}
When a sum consists of only superstars of this form, the sum is reduced to a sum of nimbers, and is thus solvable in polynomial time. As we show in the main result of our paper in section \ref{sec:reduction}, solving a sum of superstars in \cclass{NP}-hard.

\subsection{Naming Superstars}
\label{sec:superstarTerminology}

There is some historical overloading of the term ``superstar'' in two foundational CGT texts, which share an author.  In \textit{Winning Ways}, first published in 1982, superstars are defined as we use them here\cite{WinningWays:2001}.  In \textit{On Numbers and Games}, first published in 1976, the same term is used to describe specific sums of (Winning Ways) superstars.  In 2023, Silva et.\ al.\cite{SILVA2023113665} used another term, \emph{quasi-nimbers}, because they were aware of the On Numbers and Games definition, but not the one from Winning Ways.  

We do not make this choice lightly.  The terminology collision was not known until parts of this paper was presented in the Virtual Combinatorial Games Seminar\footnote{\url{https://sites.google.com/view/virtual-cgt/seminar}} in 2023.  (No one at the seminar was aware of both definitions beforehand.)  We solicited informal advice from the greater CGT community.  Based on that, we chose to use the term ``superstars'', as in Winning Ways.  Although this deviates from the first-published choice in On Numbers and Games, we are comfortable going forward with this because:
\begin{itemize}
    \item We are still using historical terminology.
    \item As pointed out by Neil McKay, ``superstar'' is nice because these games are one move above stars (nimbers) in a game tree.
    \item Only one published paper uses the term superstar in either context since the two books have been published.\footnote{Combinatorial Game Theory \cite{SiegelCGT:2013}, a text published in 2013, references the Winning Ways version in an exercise.}
\end{itemize}

This only solves the issue with superstars from Winning Ways.  In order to handle the objects described as superstars in On Numbers and Games, we propose a new term, \emph{comets}, and use that throughout.  We like this term because comets are bright celestial objects like (super)stars, but have very little mass in comparison\footnote{Although we do not provide the details of this property here, these comet positions have zero atomic weight.}.  Additionally, the alliteration of ``Conway Comets'' works nicely.

We hope that our chosen terminology will continue to be used going forward.

\iffalse
\kyle{We still need to decide what we want to call Conway superstars.  Some options:
\begin{itemize}
    \item Conway superstars
    \item Supernovas
    \item Binary star systems
    \item Binary superstars
    \item Killanova
    \item zero-weight stars
    \item dwarf stars 
    \item brown dwarves (least-weight stars)
    \item pulsars
    \item Comets (glow like stars, but are very low mass)
\end{itemize}
Feel free to add some more}
\fi

\subsection{The Game of \ruleset{Paint Can}}

\newcommand{\rsPC}{\ruleset{Paint Can}}

\rsPC\ \footnote{A playable version of \rsPC\ is available online at: \url{http://kyleburke.info/DB/combGames/paintCan.html}.} is a pleasant combinatorial game ruleset that models superstars \cite{SILVA2023113665}.

%\svenja{We need to be consistent with capitalizing the colours or not.}

\begin{definition}[\rsPC]
\rsPC\ is a combinatorial game played on stacks of bricks, each colored Red, Blue, Green, or Gray.  Each turn a player chooses a brick in one stack either of their own color or Green.  (No player can choose Gray bricks.)  The chosen brick and all bricks above it are then removed from the stack.  On top of each stack that has any non-Green bricks sits a can of green paint.  When any brick is taken from that stack, the can of paint spills, coloring all the remaining bricks in the stack Green.
\end{definition}

Under Normal Play, the last person to move on a \rsPC\ position wins.  Every superstar is equivalent to a \rsPC\ position with a single stack of bricks.  Starting with an index of zero for the bottom brick: if brick $i$ has color Blue, then $\ast i$ is only a Left option; if Red, $\ast i$ is only a Right option; if Green, $\ast i$ is both a Left and Right option; and if Gray, $\ast i$ is not an option for either player.  For example, the game \gameSet{0, \ast 2, \ast 4}{\ast, \ast 2} is equal to the stack with bricks colored (from bottom to top) Blue, Red, Green, Gray, and Blue.   
See \cref{fig:paintCan} for an example of a position with that same stack.  The entire position in the figure is equal to \gameSet{0, \ast 2, \ast 4}{\ast, \ast 2}  + $\ast 4$.

%\svenja{Add diagram, discussion about why this is interesting}

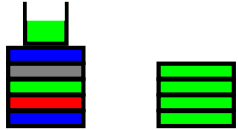
\begin{figure}[h!]
    \begin{center}
    \begin{tikzpicture}[node distance = .21cm, minimum size = .2cm, inner sep = .07cm, ultra thick]
        \node[draw, rectangle, minimum width = 1cm, black, fill = blue] (a) {};
        \node[draw, rectangle, minimum width = 1cm, black, fill = red] (b) [above of=a] {}; 
        \node[draw, rectangle, minimum width = 1cm, black, fill = green] (c) [above of=b] {}; 
        \node[draw, rectangle, minimum width = 1cm, black, fill = gray] (d) [above of=c] {}; 
        \node[draw, rectangle, minimum width = 1cm, black, fill = blue] (e) [above of=d] {}; 
        \node[rectangle] (ee) [above of=e] {};
        \node[draw, rectangle, minimum width = .55cm, minimum height = .53cm, black, fill = green] (f) [above of=ee] {}; 
        \node[draw, rectangle, minimum width = .45cm, minimum height = .28cm, white, fill = white] (g) [above of=f] {}; 

        \node[draw, rectangle, minimum width = 1cm, black, fill = green] (a2) at (2,0) {};
        \node[draw, rectangle, minimum width = 1cm, black, fill = green] (b2) [above of=a2] {}; 
        \node[draw, rectangle, minimum width = 1cm, black, fill = green] (c2) [above of=b2] {}; 
        \node[draw, rectangle, minimum width = 1cm, black, fill = green] (d2) [above of=c2] {}; 
        %\node[draw, circle, black] (c) [right of=b] {$C$};
        %\node[draw, circle, black] (d) [right of=c] {$D$};
        %\node[draw, circle, black] (e) [right of=d] {$E$};
        %\path
        %    (a) edge (b)
        %    (b) edge (c)
        %    (c) edge (d)
        %    (d) edge (e)
        %;
    \end{tikzpicture}\end{center}
    \caption{A \rsPC\ position consisting of two stacks of bricks with value \gameSet{0, \ast 2, \ast 4}{\ast, \ast 2}  + $\ast 4$.  In the leftmost stack, the Blue player may choose to remove either of the blue bricks or the green brick.  The Red player may choose to remove either the red or green brick.  Neither player may choose to remove the gray brick.   In the rightmost stack, all bricks are already green, so no can of paint is necessary.  If the Blue player removes the top brick from the left stack, the result will be a stack of four green bricks, as is in the right stack.}
    \label{fig:paintCan}
\end{figure}

Any sum of superstars can be represented as an instance of \rsPC\, with each term equivalent to a single stack of bricks.  For example, to create a position equivalent to $\gameSet{0, \ast}{0, \ast 2} + \gameSet{\ast 2}{\ast 3} + \gameSet{0, \ast, \ast 2}{\ast}$, we color bricks in each of three stacks corresponding to which player has the nimber option that matches the index of the brick (starting at the bottom with index 0).  If both players have an option to $\ast i$, then brick $i$ is green.  If only Left has an option to $\ast i$, then brick $i$ is blue.  If only Right has an option to $\ast i$, then brick $i$ is red.  If neither player has an option to $\ast i$, then brick $i$ is gray, though we do not include gray boxes for $i$ higher than the nimbers in either option.   See \cref{fig:paintCanSum} for a position equal to the prior sum of superstars.  Thus, \rsPC\ is a ruleset where all superstars and sums of superstars occur; players need to evaluate them in order to determine which player can win.

\begin{figure}[h!]
    \begin{center}
    \begin{tikzpicture}[node distance = .21cm, minimum size = .2cm, inner sep = .07cm, ultra thick]
        \node[draw, rectangle, minimum width = 1cm, black, fill = green] (a) {};
        \node[draw, rectangle, minimum width = 1cm, black, fill = blue] (b) [above of=a] {}; 
        \node[draw, rectangle, minimum width = 1cm, black, fill = red] (c) [above of=b] {}; 
        \node[rectangle] (ee) [above of=c] {};
        \node[draw, rectangle, minimum width = .55cm, minimum height = .53cm, black, fill = green] (f) [above of=ee] {}; 
        \node[draw, rectangle, minimum width = .45cm, minimum height = .28cm, white, fill = white] (g) [above of=f] {}; 

        \node[draw, rectangle, minimum width = 1cm, black, fill = gray] (a2) at (2,0) {};
        \node[draw, rectangle, minimum width = 1cm, black, fill = gray] (b2) [above of=a2] {}; 
        \node[draw, rectangle, minimum width = 1cm, black, fill = blue] (c2) [above of=b2] {}; 
        \node[draw, rectangle, minimum width = 1cm, black, fill = red] (d2) [above of=c2] {}; 
        \node[rectangle] (ee2) [above of=d2] {};
        \node[draw, rectangle, minimum width = .55cm, minimum height = .53cm, black, fill = green] (f2) [above of=ee2] {}; 
        \node[draw, rectangle, minimum width = .45cm, minimum height = .28cm, white, fill = white] (g2) [above of=f2] {}; 

        \node[draw, rectangle, minimum width = 1cm, black, fill = blue] (a3) at (4,0) {};
        \node[draw, rectangle, minimum width = 1cm, black, fill = green] (b3) [above of=a3] {}; 
        \node[draw, rectangle, minimum width = 1cm, black, fill = blue] (c3) [above of=b3] {}; 
        \node[rectangle] (ee3) [above of=c3] {};
        \node[draw, rectangle, minimum width = .55cm, minimum height = .53cm, black, fill = green] (f3) [above of=ee3] {}; 
        \node[draw, rectangle, minimum width = .45cm, minimum height = .28cm, white, fill = white] (g3) [above of=f3] {}; 
        %\node[draw, circle, black] (c) [right of=b] {$C$};
        %\node[draw, circle, black] (d) [right of=c] {$D$};
        %\node[draw, circle, black] (e) [right of=d] {$E$};
        %\path
        %    (a) edge (b)
        %    (b) edge (c)
        %    (c) edge (d)
        %    (d) edge (e)
        %;
    \end{tikzpicture}\end{center}
    \caption{A \rsPC\ position equal to $\gameSet{0, \ast 1}{0, \ast 2} + \gameSet{\ast 2}{\ast 3} + \gameSet{0, \ast, \ast 2}{\ast}$.}
    \label{fig:paintCanSum}
\end{figure}

%\kyle{Diagram should show Paint Can instances and also maybe a table of values.}

\section{From Bits to Superstars: Hardness Reduction}\label{sec:reduction}

In order to show that sums of superstars (and \ruleset{Paint Can}) are \cclass{NP}-hard, we need to introduce some additional computational problems.

\textsc{XOR-SAT} \cite{SchaeferSATDichotomy} is a classical logical satisfiability problem consisting of a conjuction of clauses of the XOR of boolean literals.  That is to say, it takes this form: $(x_i \oplus x_j \oplus \cdots \oplus x_k) \wedge (x_l \oplus \cdots \oplus x_p) \wedge \cdots \wedge (x_q \oplus \cdots \oplus x_r)$. It is known that \textsc{XOR-SAT} is polynomial-time solvable \cite{SchaeferSATDichotomy}. 

Our next problem, which is \cclass{NP}-hard, is motivated by \textsc{XOR-SAT}.  It uses \emph{multi-state variables}, which can be assigned to one of many states instead of just True and False.  Each literal of a variable is labelled with one of those states (e.g. $x_{a, s_i}$) and is only true if the variable is assigned to that state.  More formally, let $x_a$ be a \emph{multi-state} variable with possible states $s_1, s_2, \ldots, s_i, \ldots, s_k$, then for all states $s_i$ of $x_a$ we have

\[x_{a, s_i} = 
\begin{cases}
    \text{True} & \text{if }x_a \text{ is set to } s_i,\\
    \text{False} & \text{if } x_a \text{ is set to } s_j \text{ and } j \neq i.
\end{cases}\]

Figure \ref{fig:multistateVariable} displays a multi-state variable.

\begin{figure}[h!]
    \begin{center}
    \begin{tikzpicture}[node distance = .21cm, minimum size = .2cm, inner sep = .07cm, ultra thick]
        \node[draw, circle, minimum width = 4cm, black, fill=green!50] (x) [label = above:{$x_a$}]{};
        \node[] (atX) {};
        \node[draw, circle, minimum width = 1cm, black, fill=white] (s1) [above = .5cm of atX] {$s_1$};
        \node[draw, circle, minimum width = 1cm, black, fill=green!50] (s2) [right = .5cm of atX] {$s_2$};
        \node[draw, circle, minimum width = 1cm, black, fill = blue!30] (s3) [below = .5cm of atX] {$s_3$};
        \node[draw, circle, minimum width = 1cm, black, fill=red!50] (s4) [left = .5cm of atX] {$s_4$};
    \end{tikzpicture}\end{center}
    \caption{Multi-state variable $x_a$ with four possible states: $s_1$, $s_2$, $s_3$, $s_4$.  The overall color indicates that the chosen state is $s_2$.}
    \label{fig:multistateVariable}
\end{figure}
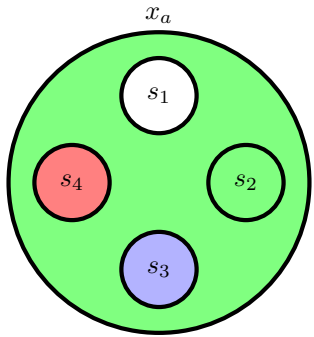

 %Removed because we turned it into a puzzle problem instead of a game
\begin{definition}
\ruleset{Multistate XOR-SAT} is a ruleset where a position is a conjunction of clauses consisting of the XOR of multi-state literals instead of boolean literals.  In other words, the clauses are of the form \((x_{i, s_i} \oplus \cdots \oplus x_{j, s_j})\). %where $s$ represents the state of the game, which can be any one of the pre-defined states. 
Variables are divided between the two players, $X$, and $Y$, and clauses may contain variables from both players, e.g. $(x_{i, s_i} \oplus y_{j, s_j} \oplus \cdots )$.  On their turn, the current player selects one of their unassigned variables and picks a state to assign it to.  Once both players have assigned all variables, $X$ wins if the formula is true, and $Y$ wins if the formula is false.  If both players have the same number of variables, then we call it \textit{Equal-Partitioned Multistate-XORSAT}, or EPMX.
\end{definition}

%changed this because we went back to the ruleset instead of the solitaire-version
\newcommand{\mxor}{\epmx}

\newcommand{\epmx}{%
\textsc{EPMX}\xspace}

\iffalse
\begin{definition}
\textsc{Multistate XOR-SAT} (\mxor) is the satisfiability problem where the formula is a conjunction of clauses consisting of the XOR of multi-state literals instead of boolean literals.  In other words, the clauses are of the form \((x_{i, s_i} \oplus \cdots \oplus x_{j, s_j})\). %where $s$ represents the state of the game, which can be any one of the pre-defined states. 
\end{definition}
\fi

%To prove \cclass{NP}-hardness, we can consider the case where the $Y$ player has no meaningful moves; there are no $y$ variables in the formula.  
We will show that \mxor is \cclass{NP}-hard after we show the reduction from \mxor to a sum of super stars (AKA \ruleset{Paint Can}).

To reduce from \mxor, we must first discuss elementary strategies in a sum of superstars. Do aid in this, we partition superstars into six classes:
\begin{itemize}
    \item nimbers,
    \item \emph{no-0}: neither player has 0 as one of their options,
    \item \emph{left-0}: only Left has a move to 0,
    \item \emph{right-0}: only Right has a move to 0,
    \item \emph{both-0}: both players have moves to 0, and
    \item \emph{one-sided}: one player has no options while the other player does.
\end{itemize}

First, an observation that follows directly from \cref{thm:superstarSimplified}:

\begin{corollary}[No-0 games]\label{no-0 games}
    A No-0 game has value 0.
\end{corollary}

Then, we will prove the following lemma:

\begin{lemma}[0 game win]\label{win strat}
Consider a sum of superstars with no both-0 games nor one-sided games. If at the start of the Left player's turn there are more left-0 games than right-0 games, then Left wins. Similarly, if at the start of the Right player's turn there are more right-0 games than left-0 games, then Right wins.
\end{lemma}

%\kyle{I think this lemma is from \cite{SILVA2023113665}.  We should probably just cite that instead of reproving it, right?}

It is possible to prove \cref{win strat} using atomic weights and the two-ahead rule, as shown in \cite{SILVA2023113665}.  We provide the following proof that avoids use of atomic weights.

\begin{proof}
We will call the winning player $A$ and the losing player $B$, so there are more $A$-0 games than $B$-0. We will prescribe the following algorithm for $A$ to win: they can ``eliminate'' 0s in the $B$-0 games by making a move on the game (by choosing one of their nimber options arbitrarily). They repeatedly do this until no $B$-0 games remain. Now, after $B$'s following turn, the remaining games include at least one $A$-0 game, some nimbers (maybe none), and some no-0 games (maybe none). At this point, $A$ should avoid playing $A$-0 games until there are only $A$-0 games remaining, or there is exactly one $A$-0 game left (along with the other types of games), whatever comes first.

If there are only $A$-0 games remaining, then for the first of those games, $A$ can just bring a game to 0, and then if there are any games left, $B$ has to make one into a nimber, which $A$ can just bring to 0. This will repeat until the last $A$-0 game is taken this way, in which case $A$ wins.

If there is exactly one $A$-0 game (along with potential no-0 games and nimbers), then $A$ can identify the value of the sum of everything but the single game by XORing the nimbers (by observation \ref{no-0 games}, the others have value 0). If the nim-sum is 0, then $A$ may take the move to 0 and thus wins the game. Otherwise, they can bring the nim-sum to 0 and inductively keep it so until $B$ plays on $A$-0, bringing the game to a non-zero nim sum, which $A$ can then win from.
\end{proof}

With this lemma, we can prove the following theorem:

\begin{theorem}
There exists a polynomial-time reduction from Equal-Partitioned Multistate-XORSAT (\epmx) to Sum of superstars, such that if True wins going first on \epmx, then the outcome class of Sum of superstars is $\mathcal{L}$ or $\mathcal{P}$  (i.e. Left wins going second).
\end{theorem}

\begin{proof}
Let $X$ be the player whose goal is to make the formula true in \epmx, $Y$ be the player whose goal is to make it false, and $m$ be the number of clauses.  We assume that the \epmx formula contains literals for each state of each variable.  (If it doesn't, we can create a dummy clause that will always be true for each missing variable missing a state.  That clause contains one copy of each state that variable can have.) 

%\svenja{Indices need cleaning}
We will use the following construction: 
First, we will assign the $t^{th}$ clause an identity $z_t = 2^{t}$. 
In other words, we use power of twos $\{1, 2, 4, \dots, 2^{m-1}\}$ to identify clauses. 
For the $i^{th}$ variable assigned to $X$, we will create a Right-0 game which we will call $x_i$, and for the $i^{th}$ variable assigned to $Y$, we will create a Left-0 game we will call $y_i$. Right's options for each $x_i$ contain only a single option of 0, and similarly, Left's options for $y_i$ contain only 0. Then, for each possible state of the $i^{th}$ variable for $X$, there will be a Left option in $x_i$ to a nimber whose value is, for each clause that contains the variable at that state, the sum of their corresponding identity values (i.e., if the $t^{th}$ clause is involved, then $z_t$ is included in the sum). Similarly, for each state of variable $y_i$ for $Y$, there will be a Right option in $y_i$ to a nimber with value equal to the sum of the corresponding identity values of the involved clauses. In addition, for each game $x_j$, there will be an additional option of $\ast 2^{m + j}$.  (The $y_i$ positions do not have this extra option.) %In addition, the game $f_1$ will be the game where Left has a move to both $\ast$ and $\ast 2^m$, and Right only has a move to 0. %The game $f_2$ is the game where Right has a move to $\ast$ and Left only has a move to 0. 

The game we consider is then %\svenja{Bounds may need changing}
\[ G = x_0+\ldots+ x_k+y_0+\ldots+y_\ell+\ast (2^m - 1).\]

For example, the position $( x_{0, a} \oplus x_{1, a} \oplus y_{0, a} \oplus y_{1, c}) \wedge (x_{0, b} \oplus x_{1, a} \oplus y_{0, b} \oplus y_{1, a}) \wedge (x_{1,a} \oplus x_{1,b}) \wedge (y_{1,a} \oplus y_{1,b} \oplus y_{1,c})$ with states $x_0: \{a, b\}$, $x_1: \{a, b\}$, $y_0: \{a, b\}$, $y_1: \{a, b, c\}$ reduces to $x_0 + x_1 + y_0 + y_1 + \ast 15$, where
\begin{itemize}
    \item $x_0 = \gameSet{\underbrace{\ast}_{a}, \underbrace{\ast 2}_b, \underbrace{\ast 16}_{2^{m+0}}}{0}$
    \item $x_1 = \gameSet{\underbrace{\ast 7}_{a}, \underbrace{\ast 4}_b, \underbrace{\ast 32}_{2^{m+1}}}{0}$
    \item $y_0 = \gameSet{0}{\underbrace{\ast}_a, \underbrace{\ast 2}_b}$
    \item $y_1 = \gameSet{0}{\underbrace{\ast 10}_a, \underbrace{\ast 8}_b, \underbrace{\ast 9}_c}$
    %\item $f_2 = \gameSet{0}{\ast}$
\end{itemize}
In the example above, the identity of these four clauses are respectively, $1$, $2$, $4$, and $8$.

Now we demonstrate the correctness of the reduction. As mentioned in the theorem statement, the union of $\mathcal{L}$ and $\mathcal{P}$ is equivalent to proving that Left wins going second.  We will show that the game should progress by alternating moves of Right playing on left-0s and Left playing on right-0s.  Since the options of those components are all nimbers, each of these plays changes the whole game by removing that component and modifying the nimber term.  If this pattern is followed, then after Left plays on the final right-0, they win if and only if the nimber term has been reduced to zero. Otherwise, Right can bring the nin-sum to 0, and then win through following the nim strategy.

If both continue to hold to that pattern, at the beginning of each of Left's turns, there is one more right-0 component than left-0.  Thus, Left must play on a right-0 component or they will lose by \cref{win strat}.  Right starts each turn with balanced left-0s and right-0s, but they still need to follow the pattern.  If Right deviates by playing on an $x$ game, then they will lose by \cref{win strat}.  % since Left playing on an $x$ will result in Right being ``two behind'' on their turn. 
If Right plays on the nimber term instead, this switches the roles of the players with respect to their ending conditions; now Right will win if and only if the nimber term is reduced to zero when they make their final move.  Left, however, can avoid this by playing on one of the large nimber (with value at least $2^m$) options included in any $x$.  Right doesn't have any options that can cancel out that large nimber, so the final nimber term will always be non-zero and Right can no longer win.  %By \cref{win strat}, right must only respond on a $y$ game, and must keep doing this so long as Left plays on $x$ games. Then, when Right plays the last $y$, since the resulting nim-sum will be non-zero (since there is no way for $\ast 2^{m+J}$ to be negated), Left can move the resulting nim-sum to 0.

%So, Right can only win if they start with playing a $y$ game. Then, by \cref{win strat}, Left must respond by playing on an $x$. This will repeat until Left plays on the final $x$, after which, Right will win if and only if the resulting nimsum is non-zero. In other words, the 

Following the prescribed sequence of play, Left wins if the nimber term equals zero.  Since it started at $\ast 2^m -1$, the sum of all nimber options chosen must also equal $\ast 2^m -1$.  If Left has a winning strategy in \epmx, they can play on the nim values corresponding to each variable state in their winning strategy, which must then result in a final nimsum of 0. If they do not have a winning strategy, then note that playing on the $2^{m + j}$ values can't give them a chance to win, since Right can stick with their \epmx strategy and the nimsum can't equal 0. Thus, Right can follow Y's winning strategy of assignments to result in a non-zero value.
\end{proof}

Now we will show NP-hardness for \mxor.  

\newcommand{\sat}{\textsc{3SAT}}

%need to finish updating the 5 state version
\begin{theorem}\label{theo:MXOR}
    There exists a polynomial-time reduction from \textsc{3SAT} to \mxor.
\end{theorem}
\begin{proof}
Let $X$ be the player whose goal is to make the formula true in EPMX, $Y$ be the player whose goal is to make it false. 
(In our reduction, $Y$ will not make meaningful decisions in the course of the game.)  Let $n$ be the number of variables and $m$ will be the number of clauses from the \sat\ instance.

Note that \sat\ is hard even if every variable appears only at most 3 times, at least once negated and at least once unnegated. It is also hard adding on a further restriction that there are an odd number of variables in the formula. We will assume both of these are true in this reduction.

For each clause in \sat, we will create a clause in \mxor. We will also have two separate clauses $c_x$ and $c_y$. For each variable ($x_i$) in \sat, we will create an variable ($x_i$) in \mxor\ with five states: $\{x_{i_a}, x_{i_b}, x_{i_c}, x_{i_d}, x_{i_e}\}$.  WLOG, we assume that $x_i$ appears once unnegated and twice negated, in clauses $r, s$, and $t$ in \sat\ respectively. Every state appears in $c_x$. In the \mxor\ formula, the first state ($x_{i_a}$) appears in no other clauses.  The second state, $x_{i_b}$, corresponds to the unnegated appearance, and appears in clause $r$.  The third state, $x_{i_c}$, corresponds to both negated \sat\ literals, and appears in clauses $s$ and $t$. The fourth state, $x_{i_d}$, corresponds to only the first negated \sat\ literal, and appears only in clause $s$. Finally, the fifth state, $x_{i_e}$, corresponds to only the second negated literal, and only appears in clause $t$.

In order to be an \epmx position, we need to include the same number of $y$ variables as $x$.  We can do this by creating $n$ dummy variables with two states each: $a$ and $b$. We can include all of the states of each of the variables into $c_y$. Note that since there are an odd number of $y$ variables, $c_y$ will always be true (the same is true of $c_x$). %We can either not include them in the formula at all, or double up on the same one in any clause to cancel each other out, e.g. including $(\cdots \oplus y_{i_a} \oplus y_{i_a})$ in any of the clauses.

If there is a solution to the \sat\ formula, then a solution to \mxor\ can be constructed by iterating over the variables.  For each $x_i$, if the \sat\ assignment is true, then we choose $x_{i_b}$, unless the \mxor\ clause $r$ has already been satisfied, in which case we choose $x_{i_a}$.  If $x_i$ is assigned to false, then we select the correct choice of $a$, $c$, $d$, and $e$, depending on which of clauses $r$ and $t$ have already been satisfied.

The inverse direction is simpler: an assignment of $x_{i_a}$ means it doesn't matter what we pick.  $x_{i_b}$ means $x_i$ must be true to satisfy the \sat\ formula, and any of the others means it must be assigned to false.
\end{proof}

%We will add a further $n$ clauses to the formula, each corresponding to a $Y$ variable. Each clause contains all 5 states of the corresponding $Y$ variable, and nothing else.

%Now note that $X$ has a winning strategy going first if and only if there is a solution that satisfies the 3SAT. 
%There is a solution to the \sat\ instance if and only if there is a solution to the new \mxor\ instance.  If there is a solution for this \mxor, there exists a solution for the \sat\ by selecting assigning the variables to true if $X$ assigns the variable to the first or second states, and otherwise to false.

%If there is a solution to the \sat\ instance, then if the assignment is true, then select state two for the variable, unless that clause is already satisfied, in which case assign it to state one. If the variable is false, then assign it to state 3, unless exactly one of the clauses is already covered, then do state 4 or 5, or if both clauses are already covered, put it in state 1.

%Comets, as we defined them earlier, are each a sum of multiple superstars, a single nimber, and either $\uparrow$ or $\downarrow$ \cite{ONAG:2000}.  We continue by showing that sums of comets are also intractible to solve.  

As explained in \cref{sec:superstarTerminology}, we introduce the term comet to refer to the objects called superstars in \emph{On Numbers And Games}\cite{ONAG:2000}.  %Each comet is the sum of a superstar, $\ast$, and either $\uparrow = \gameSet{0}{\ast}$ or $\downarrow = \gameSet{\ast}{0}$  ($\downarrow + \uparrow = 0$) depending on the value on the superstar. %so that the sum of the three components has zero atomic weight.  
Each superstar has an associated comet.  We do not provide a full explanation of all cases of comets here.  However, if the superstar, $S$, is a left-0, then the comet will be $\downarrow + \ast + S$, where $\downarrow= \gameSet{\ast}{0}$. If $S$ is a right-0, then it will be $\uparrow + \ast + S$ where $\uparrow = \gameSet{0}{\ast}$. Finally if $S$ is a nimber, then $S$ is its own comet \cite{ONAG:2000, SILVA2023113665}. Note that $\downarrow + \uparrow = 0$.

\begin{corollary}
A sum of comets is \cclass{NP}-hard.
\end{corollary}

\begin{proof}
    We can express our sums of superstars and a nimber resulting from the reduction from \epmx as a sum of comets.  If we replace each $x_i$ with the comet $\uparrow + \ast + x_i$ and each $y_i$ with the comet $\downarrow + \ast + y_i$, then sum all of those comets, all the $\uparrow$ and $\downarrow$ components will cancel out and all of the individual $\ast$ will cancel each other out. Thus the game is the same sum of superstars.
\end{proof}

A minor note is that comets have something known as atomic weight 0, which also gives the result that sums of atomic weight 0 games can be NP-hard.

\section{From Logic to Board Game: \ruleset{Blackout}}\label{sec:lightswitching}

In this section, we use a simplified version of the logical game that appeared in our earlier analysis to design a two-player board game, which we call \ruleset{Blackout}.
The board of this game contains an array of light bulbs and two sets of switches, one above the light bulbs and one below. Two players, denoted by AllOff (Left) and OneOn (Right), each control one set of switches.

\begin{figure}[!ht]
    \centering
    \includegraphics[scale=0.25]{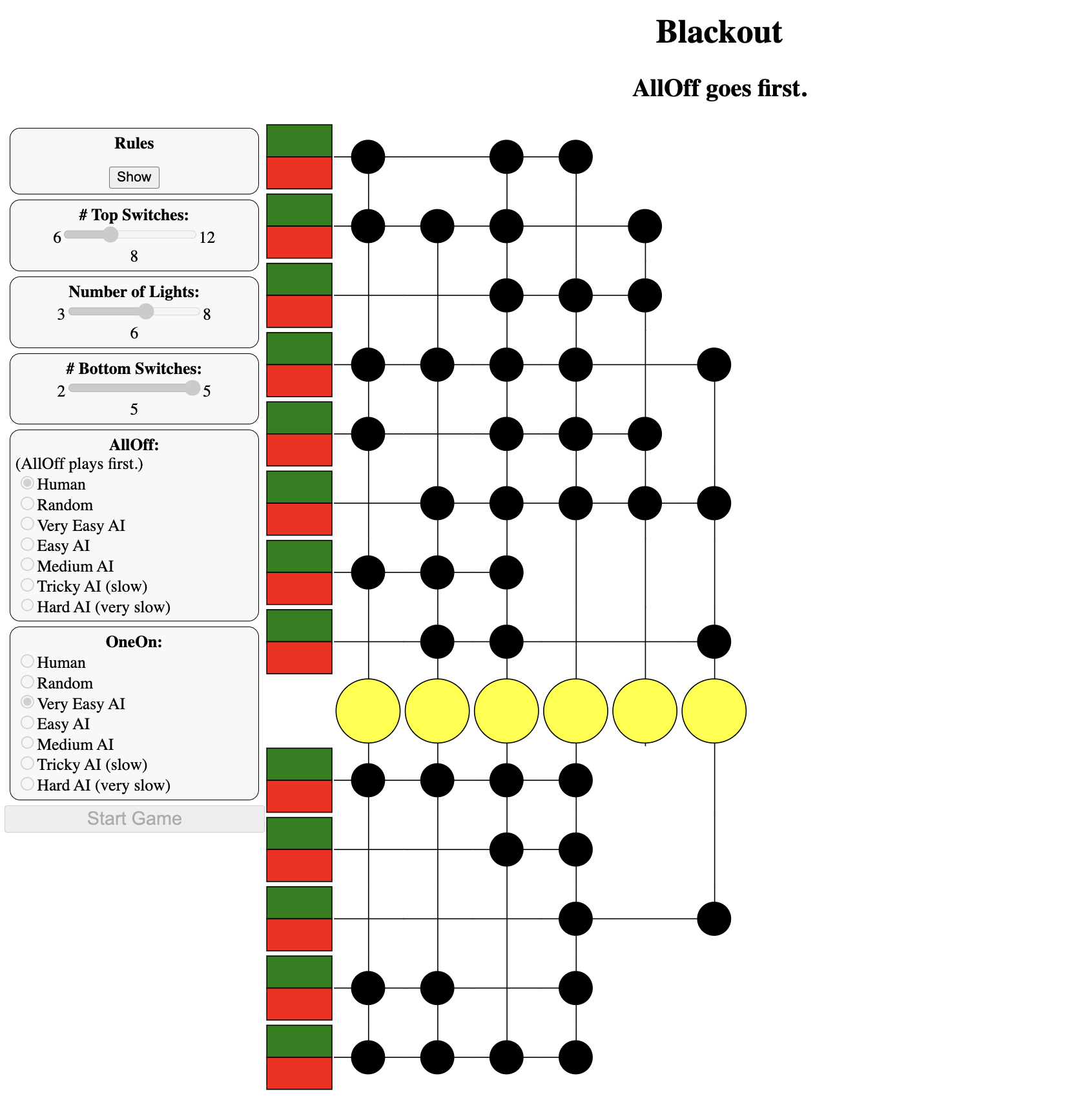}
    \caption{An Example Game of Blackout}
    \label{fig:blackout}
\end{figure}

The AllOff player wins if all lights are off at the end of the game. They control the switches at the top of the board. The OneOn player wins if at least one light is on in the end. They use the switches at the bottom of the board. On their turn, a player turns one of their unchosen switches off (red) or on (green). If they turn the switch on, then all of the lights that switch is connected to get toggled (off becomes on and on becomes off). The OneOn player has an easier objective, so they may have fewer switches. Once all bottom switches are played, they can pass as long as all the lights are not out. During these final turns, the AllOff player is searching for configuration of their remaining switches to turn all lights off in order to win.

%Each switch controls a subset of these light bulbs.  Activating a switch toggles all the lights connected to that switch. On their turn, a player selects one of their unselected switches and decides whether to activate it or deactivate it.  During the game, each switch can only be selected once, so the chosen setting is permanent. 

%an elegant real-world board game. 
%In this Maker-Breaker style game, % player AllOff is trying to get all lights to be off and player OneOn is trying to get at least one light on in the end.
\ruleset{Blackout} is in the Maker-Breaker style, but it differs from traditional Maker-Breaker games in one crucial aspect:  Even if AllOff turns all lights off at some point during the game, if OneOn still has some unflicked switches, then the game will continue until all switches for both players have been selected.
I.e., \ruleset{Blackout} could potentially continue even after all light bulbs have been turned off, while Maker-Breaker games end as soon as the desired structure has been formed (or continue until there are no moves).

\cref{fig:blackout} is an illustration from our two-dimensional board layout in our Web implementation, which can be found at the following link: 
\url{http://kyleburke.info/DB/combGames/blackout.html}.

\subsection{\ruleset{Blackout} Ruleset Formalities}

In this subsection, we discuss the mathematical representation of positions and rules in \ruleset{Blackout} to set up our computational analysis. 
%\ruleset{Blackout} is 
%partizan 
%In this Maker-Breaker style game with two players, Left is called  AllOff and Right is called OneOn, We first provide the mathematical details of this game in a way that will lead to   our practical board game design.
For \ruleset{Blackout} games with $p$ lights, a position has three components
$P = (L,\alloffset,\oneonset)$:
\begin{itemize}
  \item [(a)] $L$ is a Boolean vector in $\{0,1\}^p$, indicating whether each of the lights is off or on.
  That is, $L(i) = 1$ denotes that the $i^{th}$ light is on.
  \item [(b)]
  $\alloffset$ is a Boolean matrix with  $p$ columns in which every row has at least one 1.
  Each row represents a switch that AllOff can use. 
  If the switch controls the $i^{th}$ bulb, then its $i^{th}$ entry equals to 1.
  \item [(c)] $\oneonset$ is a Boolean matrix with  $p$ columns, in which every row has at least one 1.
  Each row represents a switch that player OneOn can use. 
  If the switch controls the $i^{th}$ bulb, then its $i^{th}$ entry equals to 1.
\end{itemize}
 
In this position, %in order to conform with Normal Play win conditions, 
we define the options as follows:

\begin{itemize}
    \item AllOff has two options for each row in $\alloffset$.  (If the matrix is empty, then they have no options remaining.)  For their turn, AllOff selects a row
$r$ and a binary action $\alpha \in \{0,1\}$.
If the action is $\alpha = 0$, then the position is moved to 
$(L,\alloffset_{-r},\oneonset)$, where $\alloffset_{-r}$ denotes the Boolean matrix obtained from $\alloffset$ by removing its $r^{th}$ row.
If the action is $\alpha = 1$,  then let $L'$ be the entry-wise exclusive-or of $L$ and the $r^{th}$ row of $\alloffset$, and the position is moved to $(L',\alloffset_{-r},\oneonset)$.
Note $L'$ represents the result when player AllOff activates its $r^{th}$ switch.
    \item OneOn's options depend on two cases: (1) If $\oneonset$ is not an empty matrix, then OneOn can select a row
$s$ and a binary action $\beta \in \{0,1\}$.
If the action is $\beta = 0$, then the position is moved to 
$(L,\alloffset,\oneonset_{-s})$, where $\oneonset_{-s}$ denotes the Boolean matrix obtained from $\oneonset$ by removing its $s^{th}$ row.
If the action is $\beta = 1$, 
then let $L'$ be the entry-wise exclusive-or of $L$ and the $s^{th}$ row of $\oneonset$,
  and position is moved to  $(L',\alloffset,\oneonset_{-s})$.
Note $L'$ represents the result when player OneOn flicks on its $s^{th}$ switch.  (2) If $\oneonset$ is an empty matrix (i.e., player OneOn has no more switches to flick),
then OneOn can make a pass move so long as there is at least one light on in $L$.  If the lights are all out, then they have no options available.\footnote{In order to prevent OneOn from making unbounded passes in the context of a game sum, they should have a maximum number of passes at the beginning equal to the difference in heights of the matrices, height($\alloffset$) $-$ height($\oneonset$).})
\end{itemize}

\subsection{The Intractability of \ruleset{Blackout}}

We now analyze the complexity of \ruleset{Blackout} and prove the following intractability result.

\begin{theorem}[Intractability of \ruleset{Blackout}]
 Deciding whether or not player {\em ALLOFF} has a winning strategy at a given \ruleset{Blackout position} is NP-hard. 
\end{theorem}
\begin{proof}

We begin by defining two decision problems, \ruleset{Set Cover} and \ruleset{Exact Cover}. 

 In \ruleset{Set Cover}, there is a collection $V$ containing $n$ sets $S_1, S_2, \dots, S_n$ which each contain some subset of a ground set $E = \{e_1, \ldots, e_m\}$ of $m$ elements. There is also a given integer $k$, indicating a target number of sets to choose.

 A \ruleset{Set Cover} instance is feasible if there exists a selection of $k$ sets in $V$ such that every element in $E$ is in at least one of the selected $k$ sets.
 We call such selection a \emph{cover}.
A cover is \textit{exact} if for each element $e \in E$, $e$ appears exactly once in the selected sets.
\ruleset{Exact Cover} determines whether the input has an exact cover.

In our desired problem of \ruleset{Pure Set Cover}, we want to have a promise that if there is a \ruleset{Set Cover} of size $k$, then there is also a \ruleset{Exact Cover} of size $k$.\footnote{We thank Neal Young for this idea.}.

Next, we note that \ruleset{Set Cover} is NP-complete even if there are only three elements in each set in $V$ \cite{GareyJohnsonStockmeyer}.
% %Some simplified NP-complete problems, M. R. Garey, D. S. Johnson, and L. Stockmeyer

Now we can reduce from \ruleset{Set Cover} with three elements in each set. For our reduction, we will enrich the input by adding new sets for each subset of those sets.  E.g., if $S_1 = \{e_1, e_2, e_3 \}$, then we include the six sets $\{e_1\}$, $\{e_2\}$, $\{e_3\}$, $\{e_1, e_2\}$, $\{e_1, e_3\}$, and $\{e_2, e_3\}$ in our new collection $V$ as well.  This enforces the promise, because if there is a set cover of size $k$, for each overlap, for one of the overlapping sets, one can instead choose to select a subset that doesn't overlap.  We can repeat this for all overlapping sets without changing $k$.

In other words, with this enrichment, we have proved that \ruleset{Pure Set Cover} is also \cclass{NP}-complete.

%\end{proposition}

To reduce from \ruleset{Pure Set Cover} to \ruleset{Blackout}, we create a light switch for the ALLOff player for each set in the \ruleset{Pure Set Cover} and $n - k$ light switches for the OneOn-player. We create a light for each element in $E$. 

The switches are set as the following: 
\begin{itemize}
    \item {\bf (AllOff)}: For the $n$ light switches controlled by the AllOff player, we connect them to the lights corresponding to the elements contained in the sets, one for each set.
    \item {\bf (OneOn)}: We connect $n-k-1$ of the light switches controlled by the OneOn-player only to the lights corresponding to the elements in  $S_1$, which the AllOff player also has a switch for.  Finally, the OneOn-player's last light switch is connected to all of the lights.
\end{itemize}

All lights begin in the on state and OneOn goes first.

We claim that the AllOff-player has a winning strategy if and only if there exists a Set Cover in the \textsc{Pure Set Cover}.

The OneOn player has $n - k - 1$ switches for $S_1$, $n - k - 2$ of which are redundant, so they should start by playing all of these, without the setting changing the outcome of the game.  

If there is a working $k$-sized cover, AllOff will spend their first $n-k-2$ turns choosing to turn off switches that are not in their pure cover and not $S_1$.  OneOn, seeing that AllOff has a switch to negate theirs for $S_1$, should play that again.  If $S_1$ is part of the cover, then AllOff can choose the setting that shuts them off.  Otherwise, they should choose to turn them on.  OneOn will then choose to leave all lights on (otherwise AllOff can win immediately).  Now AllOff has $k$ turns to flip all switches in their cover of size $k$ to turn all lights off.  If OneOn decides to activate the switch connected to all lights earlier, they should choose to leave them all on, in which case AllOff just has extra turns to put their cover to work.  No matter what, AllOff will win.

%If there exists a Exact Cover of size $k$ in the \textsc{Pure Set Cover}, then AllOff's winning strategy is to select switches corresponding to every set except the $k$-sized Exact Cover.  Then, the OneOn player will be out of moves, and all of the lights will still be on (since the other player wouldn't have flipped them to off, which would cause the AllOff-player to immediately win). Then, the AllOff player can simply play the remaining sets to win the game.

Then, if there exists no Exact Cover of size $k$, then there exists no set cover of size $k$.  OneOn can win by saving their all-lights-switch for their last move.  Since there is no set cover of size $k$, there must be at least one light not covered by the AllOff player's remaining switches. The OneOn player can either flip or not flip the final switch to make sure that light is on and win the game.
\end{proof}

%\subsection*{A Polynomial-Time Solution for Balanced Position}

\section{Conclusion}

Going beyond winnability,  Sprague-Grundy Theory \cite{Sprague:1936,Grundy:1939}
introduced the first notion of \emph{game values} and \emph{game algebra} for combinatorial games, which was later expanded to partizan games \cite{WinningWays:2001,ONAG:2000}.
It demonstrated that every impartial game can be mathematically reduced to a single-pile of \ruleset{Nim} in the algebra of disjunctive sums.
Because the nim sum can be computed in linear-time in the size of the binary representations of the summands \cite{Bouton:1901}, Sprague-Grundy Theory further captures the computational benefit of knowing game values of impartial games, rather than just their game rules or winnability \cite{BurkeFerlandTengGrundy}. 

By showing that the winnability of the sum of superstars and comets are intractable, our result highlights the fundamental subtlety of tepid partizan game values. 
It takes a significant step beyond Morris' \cite{morris1981playing} classical intractability result by demonstrating that intractability happens just one step above nimbers (stars).
This hardness result implies that the Bouton-like result for \ruleset{Nim} \cite{Bouton:1901}  is unlikely for superstars and comets.

Part of our proof has also inspired the design of a board game, \ruleset{Blackout}, which enjoys intriguing complexity
%We present a sharp complexity characterization of this game 
  based on the shape of game boards:
  On the one hand, if both players have the same number of switches, which we refer to as the balanced case, then the game can be solved in polynomial time \footnote{This statement follows from the basic idea in \cite{SchaeferSATDichotomy}.}.
On the other hand, if players have a different number of switches, then the game is intractable in general.

Whereas the sum of superstars has been shown to be \cclass{NP}-hard to solve, its precise complexity remains open. 
It can be shown that the outcome class of a sum of superstars can be computed in polynomial space in the number of bits representing the sum.\footnote{One of the proofs was given by Aaron Siegel.}
Therefore, as part of the next step in our future research, we would like to settle the following conjecture.

\begin{conjecture}
A sum of superstars, and hence \ruleset{Paint Can}, is \cclass{PSPACE}-complete. 
\end{conjecture}

It is worth noting that it seems challenging to naturally extend the current setup of our intractability proof for this conjecture. Reducing directly from QSAT to EPMX seems fruitless since there is no clear way to ``punish'' player Y from covering a variable multiple times.
For any reduction, the fundamental difficulty lies in the strong asymmetry between players X and Y in that Y is just too powerful with options.
In other words, 
%outside of QSAT reductions,
as soon as EPMX allows Y to be a ``real'' decision maker---as opposed to the construction in Theorem \ref{theo:MXOR}---the game shifts dramatically in Y's favor, and indeed, it is difficult to even find complicated positions where X wins with a non-trival Y player. As it stands, either a different approach is needed, or a very clever reduction to EPMX is required.

%\emph{Open Problem}: Is \ruleset{Paint Can} \cclass{PSPACE}-complete?

\bibliographystyle{plainurl}

% \bibliography{paithan}  

\end{document}